\documentclass[a4paper,10pt,twoside,reqno,nonamelimits]{article}
\usepackage{a4wide}
\usepackage[colorlinks=true,urlcolor=blue,linkcolor=blue,citecolor=blue]{hyperref}
\usepackage{newcent}         % PSNFSS

\usepackage{amsmath}
\usepackage{amssymb}
\usepackage{amsfonts}
\usepackage{amscd}
\usepackage{amsthm}
\usepackage{amsxtra}
\usepackage{MnSymbol}
\usepackage{mathrsfs}
\usepackage{nicefrac}
\usepackage{graphicx}%\graphicspath{ {image/} }
\usepackage{xcolor}
\usepackage{tikz}
\usepackage{enumitem}
\setlist[enumerate]{leftmargin=2em,itemindent=0em, labelindent=0pt,labelwidth=1.5em,labelsep=.5em, align=left, noitemsep}
\newlist{txtenum}{enumerate}{1}
\setlist[txtenum]{leftmargin=0em,itemindent=1.5em, labelindent=0pt,labelwidth=1em,labelsep=.5em, align=left}
\usepackage[norelsize,boxed,noend,linesnumbered]{algorithm2e}\DontPrintSemicolon
\SetKwFor{RepeatTimes}{repeat}{times}{endrepeat}
\RestyleAlgo{ruled}
\usepackage{datetime}

\theoremstyle{plain}
\newtheorem{theorem}{Theorem}
\newtheorem*{theorem*}{Theorem}

\newtheorem{proposition}[theorem]{Proposition}
\newtheorem*{proposition*}{Proposition}

\newtheorem*{corollary*}{Corollary}

\newtheorem{lemma}[theorem]{Lemma}
\newtheorem*{lemma*}{Lemma}

\newtheorem*{observation*}{Observation}

\newtheorem*{conjecture*}{Conjecture}

\newtheorem*{question*}{Question}

\newtheorem*{questions*}{Questions}

\newtheorem*{problem*}{Problem}

\newtheorem*{problems*}{Problems}

\newtheorem*{openproblem*}{Open Problem}

\theoremstyle{definition}

\newtheorem*{definition*}{Definition}
\newtheorem{example}[theorem]{Example}
\newtheorem*{example*}{Example}

\newtheorem*{exercise*}{Exercise}

\newtheorem{remark}[theorem]{Remark}
\newtheorem*{remark*}{Remark}

\newtheorem*{remarks*}{Remarks}

\theoremstyle{remark}

\newtheorem*{claim*}{Claim}

\newcommand{\subclass}[1]{}

\newcommand{\enumTi}[1]{\renewcommand{\theenumi}{#1}}

\newcommand{\alphenumi}{\enumTi{\alph{enumi}}}

\newcommand{\romenumi}{\enumTi{\roman{enumi}}}

\newlength{\hspaceforlengthglumpf}

\newcommand{\comment}[1]{\text{\footnotesize[#1]}}

\DeclareMathOperator{\tr}{tr}

\newcommand{\One}{\mathbf{1}}

\newcommand{\lt}{\left}
\newcommand{\rt}{\right}

\newcommand{\abs}[1]{{\lt\lvert{#1}\rt\rvert}}

\newcommand{\nfrac}[2]{{\nicefrac{#1}{#2}}}

\newcommand{\NN}{\mathbb{N}}

\newcommand{\RR}{\mathbb{R}}

\newcommand{\ZZ}{\mathbb{Z}}

\newcommand{\bra}[1]{{\lt< #1 \rt|}}
\newcommand{\ket}[1]{{\lt| #1 \rt>}}

\newlength{\algotabbingwidth}
\renewcommand{\paragraph}[1]{\medskip\noindent{\textsl{#1.}}}
\newcommand{\mypar}{\par\medskip\noindent}

\begin{document}
\title{Calculus on Parameterized Quantum Circuits\thanks{Supported by the Estonian Research Council, ETAG (\textit{Eesti Teadusagentuur}), through PUT Exploratory Grant \#620.}}
\author{Javier Gil Vidal$^a$ and Dirk Oliver Theis$^{a,b}$
\\[1ex]
  \small $^a$ Institute of Computer Science, University of Tartu, Estonia\\
  \small $^b$ Ketita Labs {\tiny O\"U}, Tartu, Estonia\\
  \small \texttt{javier@ut.ee}, \texttt{dotheis@}\{\texttt{ketita.com}, \texttt{ut.ee}\}%
}
%%%%%%%%%%%%%%%%%%%%%%%%%%%%%%%%%%%%%%%%%%%%%%%%%%%%%%%%%%%%%%%%%%%%%%%%%%%%%%%%%%%%%%%%%%%%%%%%%%%%%%%%%%%%%%%%%%%%%%%%%%%%%%%%%%%%%%%%%%%%%%%%%%%%%%
%%
%%
\date{Mon Dec 31 13:48:10 UTC 2018 [Compiled: \currenttime]}
%%
%%
%%%%%%%%%%%%%%%%%%%%%%%%%%%%%%%%%%%%%%%%%%%%%%%%%%%%%%%%%%%%%%%%%%%%%%%%%%%%%%%%%%%%%%%%%%%%%%%%%%%%%%%%%%%%%%%%%%%%%%%%%%%%%%%%%%%%%%%%%%%%%%%%%%%%%%
\maketitle

\begin{abstract}
  Mitarai, Negoro, Kitagawa, and Fujii proposed a type of parameterized quantum circuits, for which they gave a way to estimate derivatives wrt the parameters.  Their method of estimating derivatives uses only changes in the values of the parameters, and no other changes to the circuit; in particular no ancillas or controlled operations.  Recently, Schuld et al.\ have extended the results, but they need to revert to ancillas and controlled operations for some cases.

  In this short paper, we extend the types of MiNKiF circuits for which derivatives can be computed without ancillas or controlled operations --- at the cost of a larger number of evaluation points.  We also propose a ``training'' (i.e., optimizing the parameters) which takes advantage of our approach.
  \par\medskip%
\end{abstract}

\section{Introduction}\label{sec:intro}
In near-term quantum computing, there is considerable interest in so-called \textit{Parameterized Quantum Circuits (PQCs)}, where quantum operations are dependent on parameters which are iteratively modified to ``train'' the quantum circuit to compute a desired function.  One such approach is \textit{Quantum Circuit Learning} by Mitarai et al., \cite{Mitarai-Negoro-Kitagawa-Fujii:q-circ-learn:2018}, where real number parameters $\theta_1,\dots,\theta_m$ determine the duration for evolving the system determined by a generalized Pauli Hamiltonian, or, in simpler words, running the unitary operation $e^{-i\theta_j P_j/2}$.  The whole quantum computation then estimates the function
\begin{equation}\label{eq:qcirc-Fn}
  F\colon \RR^m \to \RR\colon \theta \mapsto \tr(M U(\theta) \rho U(\theta)^\dagger),
\end{equation}
with $U(\theta) = e^{-i\theta_m H_m} V_{m-1} \dots e^{-i\theta_2 H_2} V_1 e^{-i\theta_1 H_1}$ for Hamiltonians $H_1,\dots,H_m$ and arbitrary (fixed) unitaries $V_1,\dots,V_{m-1}$.  We refer to PQCs of this form as \textit{MiNKiF} PQCs.  In the original paper of~\cite{Mitarai-Negoro-Kitagawa-Fujii:q-circ-learn:2018}, the Hamiltonians are generalized Pauli operators (up to normalization, in our notation), but that is clearly not essential to the approach.

Mitarai et al.'s approach has proven popular due to the ostensible ease with which partial deriveatives $d/d\theta_j F(\theta)$ are available.  Ignoring the reality of near-term quantum computing devices and some amount of mathematics, one could say that evaluating a partial derivative is only twice as costly as evaluating the function~$F$ itself.  This is the result of the following observation~\cite{Mitarai-Negoro-Kitagawa-Fujii:q-circ-learn:2018}:
\begin{equation}
  \frac{d}{d\theta_j} F(\theta) = \tfrac12\lt(  F(\theta + \pi e_j /2) - F(\theta - \pi e_j /2)  \rt),
\end{equation}
where $e_j$ is the $m$-vector with a~$1$ in position~$j$ and $0$ everywhere else.
In other words, it suffices to evaluate \textbf{the same quantum circuit} --- only modified parameters --- a couple of times to evaluate the partial derivative.

Recently, Schuld et al.~\cite{Schuld-Bergholm-Gogolin-Izaac-Killoran:ana-grad:2018} have given a generalization of that result to more general Hamiltonians than generalized Paulis.  They extend Mitarai et al.'s method so that the generalized Pauli Hamiltonians can be replaced with any~$H_j$ which have two distinct eigenvalues.  Moreover, for~$H_j$'s with more than two eigenvalues, Schuld et al.\ propose a method for evaluating the gradient which need an ancilla qubit, and require to apply $e^{-itH_j}$ and unitaries derived from it controlled on that ancilla.

\mypar%
Using an ancilla comes with several drawbacks on near-term quantum computers.  The first comes from the limited qubit connectivity of near-term quantum processors: One would like to lay out the MiNKiF PQC on the quantum processor in a way which makes maximum use of the hardware capabilities, but the circuit computing the derivative requires the ancilla qubit to be transported to the locations where hardware-native controlled operations are required, which gives not only a significant overhead, but also a significant difference to the original layout.  Secondly, and perhaps more importantly, the quantum noise present in the derivative-circuit will be considerably higher than for the original circuit, and it will also be different in a way which is difficult to control for.  In other words, due to quantum noise, the derivative circuit will yield expectation values which may not be related much with the derivative not of the original PQC function.

\paragraph{Our results}
In this note, we show how the derivatives of the quantum circuit can be computed without ancilla qubits \textit{if} the eigenvalues of the Hamiltonians~$H_j$ are (known and) spaced nicely: For each~$H_j$ there exists an $\alpha$ such that all pairwise differences of eigenvalues of~$H_j/\alpha$ are integral.

We also point out how our technique can be further exploited in a coordinate descent ``training'' algorithm.

\section{Eigenvalue distances and Fourier transform}
Without loss of generality, we restrict our attention to a single parameter:
For a fixed hermitian operator (Hamiltonian) $H$ let $U(t) = e^{-i t H}$; fix a hermitian operator (observable)~$M$ and positive a hermitian operator with trace one (initial quantum state)~$\rho$.
Consider the expectation value
\begin{equation}
  f_H\colon \RR \to \RR\colon t \mapsto \tr(M U(t) \rho U(t)^\dagger).
\end{equation}

We will denote by $\lambda_j$, $j=1,\dots,n$, the eigenvalues, listed with multiplicities, of~$H$, and by $\ket{j}$, $j=1,\dots,n$, a corresponding orthonormal eigenbasis.

The assumption on the eigenvalues sketched in the introduction means that there exist integers $k_1,\dots,k_n$ with
\begin{equation}\label{eq:scale-eigvals}
  \lambda_j = \alpha k_j,
\end{equation}
so that the hermitian operator $K := H/\alpha$ has integral eigenvalues $k_1,\dots,k_n$.  Let
\begin{gather}
  V(t) := e^{-i t K}\\
  \intertext{and}
  g\colon \RR \to \RR\colon t \mapsto \tr(M V(t) \rho V(t)^\dagger);
\end{gather}
Note that
\begin{equation}\label{eq:f-in-terms-of-g}
  \begin{aligned}
    g(t) &= f(t/\alpha) \\
    f'(t) &= \alpha g'(\alpha t).
  \end{aligned}
\end{equation}
so that~$g(t)$ can be evaluated by evaluating $f(t/\alpha)$, and the derivative of~$g$ gives the derivative of~$f$.  In other words, without loss of generality, we may assume that the eigenvalues $\lambda_j$ are integral.  We do so, and drop the use of~$g$ from now on.

\begin{example}
  Consider the microwave-controlled transmon gate for superconducting architectures from \cite{Chow-Corcoles-Gambetta-Rigetti-etal:entangle-microwave:2011} (cf. \cite{Schuld-Bergholm-Gogolin-Izaac-Killoran:ana-grad:2018}):
  \begin{equation*}
    H := \sigma_x\otimes \One - b\sigma_z\otimes\sigma_x + c\One\otimes\sigma_x.
  \end{equation*}
  The eigenvalues are $\pm c \pm \sqrt{b^2+1}$.  Our technique applies if $c$ and $\sqrt{b^2+1}$ are collinear over the rational numbers.
\end{example}

We start with some easy lemmas.

\begin{lemma}
  The function~$f$ is $2\pi$-periodic.
\end{lemma}
\begin{proof}
  This follows by simply noting that, for all~$t$,
  \begin{align*}
    V(t+2\pi)
    &= \sum_{j=1}^n e^{-i(t+2\pi)k_j} \ket{j}\bra{j}                             \\
    &= \sum_{j=1}^n e^{-i t k_j} \ket{j}\bra{j}       && \comment{$k_j \in \ZZ$} \\
    &= V(t).
  \end{align*}
\end{proof}

\newcommand{\fouint}{\tfrac{1}{\sqrt{2\pi}}\int_0^{2\pi}e^{-ikt}}
Denote by
\begin{equation*}
  \hat f\colon k \mapsto \fouint f(t) \,dt
\end{equation*}
the Fourier transform of~$f$.

\begin{lemma}\label{lem:supp-of-hatf}
  The Fourier transform $\hat f$ of~$f$ is supported on $D := \{ k_i-k_j \mid i,j=1,\dots,n \}$.
\end{lemma}
\begin{proof}
  Let's do it!
  \begin{align*}
    \hat f(k)
    &= \fouint \tr(MV(t)\rho V(-t)) \,dt                 \\
    &= \tr\lt( M \cdot \fouint V(t)\rho V(-t) \,dt \rt)  &&  \\
    &= \sum_{l,j=1}^n \tr\lt( M \cdot \fouint V(t) \ket{l}\bra{l} \rho V(-t) \ket{j}\bra{j} \,dt \rt) && \\
    &= \sum_{l,j=1}^n \tr\lt( M \cdot \fouint e^{-i t k_l} \ket{l}\bra{l} \rho e^{i t k_j} \ket{j}\bra{j} \,dt \rt) && \\
    &= \sum_{l,j=1}^n \fouint e^{-i t k_l} e^{i t k_j} \,dt \tr\lt( M \cdot \ket{l}\bra{l} \rho \ket{j}\bra{j} \rt). && \\
  \end{align*}
  Now
  \begin{align*}
    \fouint e^{-i t k_l} e^{i t k_j} \,dt
    & = \frac{1}{\sqrt{2\pi}}\int_0^{2\pi}e^{it( (k_j-k_l) - k)} \\
    &= \begin{cases}
      \sqrt{2\pi},& \text{if $k_j-k_l = k$}\\
      0,&           \text{otherwise.}
    \end{cases}
  \end{align*}
  Hence, $\hat f(k)$ is zero, unless $k\in D$.
\end{proof}

The Fourier inversion theorem now gives the Fourier expansion of~$f$ and~$f'$:
\begin{align*}
  f(t) &= \tfrac{1}{\sqrt{2\pi}} \sum_{k\in D} e^{itk}\hat f(k), \\
  \intertext{and}
  f'(t) &= \tfrac{1}{\sqrt{2\pi}} \sum_{k\in D\setminus\{0\}} ik e^{itk}\hat f(k).
\end{align*}

We immediately derive the result.

\begin{proposition}\label{prop:diamD+1}
  The derivative function $f'$ can be determined by evaluating~$f$ in $S := 2(\lambda_{\max} - \lambda_{\min})/\alpha +1$ arbitrary distinct points in $\lt[0,2\pi/\alpha\rt[$.
\end{proposition}
\begin{proof}
  Let $n := \max_{d\in D} \abs{d}$, and note that $S = 2n+1$.
  Take arbitrary distinct points $t_1,\dots,t_S\in \lt[0,2\pi\rt[$, and evaluate $f(t_s)$, for $s=1,\dots,S$.

  Then solve the non-singular system of linear equations
  \begin{equation}\label{eq:fourier-matrix}
    \begin{pmatrix} f(t_1)\\ \\ \vdots \\ \\ \\ f(t_S) \end{pmatrix}
    =
    \begin{pmatrix}
      e^{-i t_1 n} & \dots & e^{-i t_1} & 1 & e^{i t_1} & \dots & e^{i t_1 n} \\
                   &       &            &   &           &       &             \\
                   &       &            &   &           &       &             \\
      \vdots       &       &            &   &           &       &  \vdots     \\
                   &       &            &   &           &       &             \\
                   &       &            &   &           &       &             \\
      e^{-i t_S n} & \dots & e^{-i t_S} & 1 & e^{i t_S} & \dots & e^{i t_S n} \\
    \end{pmatrix}
    \begin{pmatrix}
      \beta_{-n} \\ \\ \vdots \\ \beta_0 \\ \vdots \\ \\ \beta_n
    \end{pmatrix}
  \end{equation}
  for $\beta_{\centerdot}$.
  We have $\hat f(k) = \beta_k$ for all $k=-n,\dots,n$ and we can evaluate the derivative of~$f$ in~$t$ as
  \begin{equation*}
    f'(t) = \sum_{k=-n}^n k\beta_k e^{ikt}.
  \end{equation*}
\end{proof}

\subsection{Improvement for 2 and 3 eigenvalues}
For two eigenvalues, the technique in Prop.~\ref{prop:diamD+1} requires one more evaluation than the one in \cite{Mitarai-Negoro-Kitagawa-Fujii:q-circ-learn:2018,Schuld-Bergholm-Gogolin-Izaac-Killoran:ana-grad:2018}.  We now explain how the two approaches are related, and reduce the number of evaluations of~$f$ for the case of 3 equally spaced eigenvalues to~4.

As~$f$ is a real valued function, let us rewrite the Fourier expansion as a trigonometric polynomial: There are real numbers $\alpha$, $\beta_k$, $\gamma_k$, $k\in\NN:=\{1,2,3,\dots\}$, such that
\begin{equation}\label{eq:trig-poly}
  f(t)
  = \alpha
  + \sum_{k=1}^n \beta_k \sin(kt)
  + \sum_{k=1}^n \gamma_k\cos(kt).
\end{equation}

\paragraph{Two eigenvalues}
We now cast the method of \cite{Mitarai-Negoro-Kitagawa-Fujii:q-circ-learn:2018} in our context.

For two eigenvalues, we only have the terms with $k=1$ in the expansion~\eqref{eq:trig-poly}:
\begin{equation}\label{eq:lin-trig}
  f = \alpha + \beta\sin + \gamma\cos.
\end{equation}
The method of \cite{Mitarai-Negoro-Kitagawa-Fujii:q-circ-learn:2018} requires evaluating~$f$ at the points $\pi/2$ and $-\pi/2$, which gives us a non-singular system of two equations with two unknowns, $\alpha$, $\beta$.

\paragraph{Three evenly-spaced eigenvalues}
Suppose that there are three distinct eigenvalues, but with only two non-trivial differences: $D=\{0,\pm 1,\pm 2,\}$ (after normalization).
In this case we have the terms with $k=1,2$ in the expansion~\eqref{eq:trig-poly}:
\begin{equation*}
  f(t) = \alpha + \beta_1\sin(t) + \beta_2\sin(2t) + \gamma_1\cos(t) + \gamma_2\cos(2t).
\end{equation*}
Evaluating~$f$ in the four points $\pm \pi/4$, $\pm 3\pi/4$ gives us a non-singular system of 4 equations with 4 unknowns $\alpha$, $\beta_1$, $\beta_2$, $\gamma_1$, thus allowing to determine $\beta_1,2$ and hence $f'(0)$ with four evaluations of~$f$.

\paragraph{More eigenvalue differences}
It appears that this method cannot be extened to more than $2$ non-trivial differences, since is no selection of points to evaluate~$f$ in which would eliminate a variable from the system of equations.

\subsection{Sparse difference sets}
It may happen that $\abs{D} < 2\max_{d\in D}\abs{d} +1$.  In that case, the proposition above is not optimal in that it requires too many evaluations of~$f$.  We now give the ``right'' result.

\begin{theorem}\label{thm:|D|-evals}
  The Fourier transform of~$f$ can be obtained from $S := \abs{D}$ evaluations of~$f$, for example in the following ways:
  \begin{enumerate}[label=(\alph*)]
  \item\label{thm:|D|-evals:vandermonde} Fix an integer~$a$ and evaluate~$f$ in the points $t_j := 2\pi(a+j)/S$, $j=0,\dots,S-1$; or
  \item\label{thm:|D|-evals:random} Evaluate~$f$ in $\abs{D}$ points chosen independently uniformly at random from $[0,2\pi]$.
  \end{enumerate}
\end{theorem}
\begin{proof}
  \textit{\ref{thm:|D|-evals:vandermonde}.}  For $k\in D$, let $z_k := e^{2\pi i k/S}$.  Since the $z_k$, $k\in D$, are distinct, the Vandermonde matrix defined by
  \begin{align*}
    A_{j,k}  &:= (z_k)^j                 &&\text{$k\in D$, $j=0,\dots,S-1$} \\
    \intertext{is invertible, and hence, the system of linear equations}
     f(t_j) &= e^{2\pi i a/S} \sum_{k\in D} z_k^j x_k &&\text{$k\in D$, $j=0,\dots,S-1$}
  \end{align*}
  has a unique solution~$x$.  Clearly, $x_k = \hat f(k)$.

  \textit{\ref{thm:|D|-evals:random}.}  Let $D = \{k_1,\dots, k_S\}$, with $S:=\abs{D}$.  For $t\in\RR^S$, consider the $(S\times S)$-matrix $F$ defined by $F(t)_{\ell,j} := e^{i k_\ell t_j}$.  Letting $y_j := f(t_j)$, by Lemma~\ref{lem:supp-of-hatf}, the system of linear equations
  \begin{equation*}
    y = F(t) x
  \end{equation*}
  has $x_\ell = \hat f(k_\ell)$ as a solution.  We want to show that it is the only solution; in other words, that the matrix~$F(t)$ is non-singular if~$t$ is chosen uniformly at random.

  For that, consider the real analytic function $h\colon \RR^S \to \RR\colon t\mapsto \det(F(t))$, and assume that there were a set~$Z\subseteq[0,2\pi]^S$ of non-zero $S$-dimensional Lebesgue measure such that $h(Z)=0$.  Using standard techniques (see Lemma~\ref{lem:zeros-of-analytic-functions} in Appendix~\ref{apx:analytic-sets}), it can be shown that then~$h$ were identically zero on $[0,2\pi]^S$ --- which would contradict Item~\ref{thm:|D|-evals:vandermonde}.
\end{proof}

\begin{remark}
  We note that, in the simplest cases where the eigenvalues of~$H$ are $\pm \nfrac12$ (after normalizing), it is possible to obtain the coefficients in~\eqref{eq:lin-trig} directly by evaluating expectation values of circuits, using $e^{-itH} = \cos t \One -i sin t H$, and applying the tricks in \cite{Mitarai-Negoro-Kitagawa-Fujii:q-circ-learn:2018,Schuld-Bergholm-Gogolin-Izaac-Killoran:ana-grad:2018}.
\end{remark}

We refer to \cite{GilVidal:PhD:2020} for the application of Theorem~\ref{thm:|D|-evals} to the usual parameterized 2-qubit gates.

\section{Exploiting the Fourier expansion algorithmically}
It is reasonable to assume that an algorithm for ``training'' a PQC would evaluate the expectation value for the current set of parameters in every iteration of the ``training'' algorithm.  In our setting, that would imply that $f(0)$ is always known \textsl{for free,} as it were.
Even if that is not the case, spending an additional evaluation of~$f$ to obtain the whole Fourier expansion of~$f'$ might pay off in the training.

The effect is most easily explained in the case where the goal of the ``training'' is to find the ground state of the Hamiltonian~$M$ in~\eqref{eq:qcirc-Fn} through some kind of Variational Quantum Eigensolver~\cite{Peruzzo-Mcclean-Shadbolt-Yung-Zhou-Love-AspuruGuzik-Obrien:VQE:2014}.\footnote{We note, though, that the advantage persists when the quantum circuit function~$F$ \eqref{eq:qcirc-Fn} takes inputs depending on other parameters and its result is used as input to other mechanism, and the whole system has to be trained (as would be the case, e.g., in a heterogeneous quantum-classical neural network).}

Here is a sketch of a coordinate-descent algorithm for minimizing~$F$.  (No attempt has been made to make it sophisticated.)  For the sake of simplicity, we gloss over details of normalization for the eigenvalues, and assume that the eigenvalue distances of all Hamiltonians are integral.

\begin{algorithm}[H]
  \caption{Coordinate Descent for Quantum Circuits}
  \SetKw{Input}{Input}\SetKw{Output}{Returns}%
  \Input{\\%
    \quad$\bullet$ Oracle access to evaluating~$F$ given $\theta\in\RR^m$\\
    \quad$\bullet$ Starting point $\theta\in\RR^m$\\
    \quad$\bullet$ Sets $D_j$, $j=1,\dots,m$, of eigenvalue distances of Hamiltonian~$j$\\
  }%
  \Output{A point $\theta\in\RR^m$ where~$F$ attains a local minimum}\\%
  \BlankLine
  \Repeat{No improvement is made}{%
    \For{$j=1,\dots,m$}{%
      Evaluate $f(t) := F(\theta + t e_j)$ for $\abs{D_j}$ choices for~$t$ (cf. Theorem~\ref{thm:|D|-evals})\\
      Compute the Fourier transform $\hat f(k)$, $k\in D_j$ (cf. Theorem~\ref{thm:|D|-evals})
      and find the $\alpha$, $\beta_k,\gamma_k$, $k\in D_j\cap\NN$ \\
      Find the minimum of the trigonomeric polynomial~\eqref{eq:trig-poly}      \label{algstep:min}\\
      Denote the minimum by~$t_0$, and update $\theta = \theta + t_0 e_j$.
    }% for j
  }% repeat-until
  \Return{$\theta$}
\end{algorithm}

The difficulty lies, obviously in Step~\ref{algstep:min}.

\paragraph{Case of two eigenvalues}
In the case of two eigenvalues (i.e., $D = \{0,\pm1\}$ after normalization) Step~\ref{algstep:min} is trivial:  Using~\eqref{eq:lin-trig}, the minimum is one of $t_1:=\arctan(\beta/\gamma) \in \lt]-\nfrac\pi2,\nfrac\pi2\rt]$ (with $\arctan(\infty) := +\nfrac\pi/2$) or $t_2:=t_1+\pi$, depending on the signs of $\beta,\gamma$: If $\gamma\ne 0$, then $t_1$ is the maximum iff $\gamma >0$; if $\gamma=0$, then $t_1$ is the maximum iff $\beta > 0$.

\paragraph{More eigenvalues}
Now assume that there are more eigenvalues; we retain the condition that the eigenvalue distances are a integral and their greatest common divisor is~1.  Then solving Step~\ref{algstep:min} appears, at first sight, daunting.  The following is, however, well-known \cite{boyd-vandenberghe:book:2004}.

\begin{theorem}
  Let $f(t)=\sum_{k=-n}^n \zeta_k e^{-ikt}$ with $\zeta_k=\zeta_{-k}^*$ and $\zeta_0=0$.  Then the minimum of~$f(t)$ over~$t \in [0,2\pi]$ is equal to the maximum  of $-\tr(F)$ where~$F$ ranges over all $n\times n$ complex hermitian positive matrices which satisfy, for $k=1,\dots,n$,
  \begin{equation*}
    \sum_{j=k+1}^n F_{j,j-k}  = \zeta_k.
  \end{equation*}
  Moreover, the minimum can be determined by solving a semidefinit programming problem.
\end{theorem}

As semidefinite programming can be performed to any given accuracy in polynomial time --- and can be performed efficiently in practice, too --- this gives a way to solve Step~\ref{algstep:min} for any reasonably sized~$n$ to reasonabe precision.

\section{Summary and outlook}
We have demonstrated that in well-behaved cases, evaluating derivatives of expectation values of MiNKiF PQCs can be done without ancilla qubits.
The given number of evaluations prescribed by our technique for evaluating the derivative in a single point (except in the cases of 2 eigenvalues and 3 evenly spaced eigenvalues) already gives the whole Fourier expansion of the function.  We have argued that this can be used to derive coordinate-descent- based optimization algorithms which, in every step, jump to the minimum of the function in the given coordinate direction.

This work raises some obvious questions.

\paragraph{Randomness}
In this paper, we have made the assumption that the expectation value $F(\theta)$ can be ``evaluated''.  This is mathematically inaccurate: It can only be estimated as the average as the result of several runs of the quantum circuit with (identical parameter setting), and the required number of samples depends on the variance.  In view of Theorem~\ref{thm:|D|-evals}, it is an obvious question whether the accuracy of determining the Fourier amplitudes can be increased by varying the parameters between the samples.

\paragraph{Sign vs.\ Fourier expansion}
The currently proposed algorithms for ``training'' PQCs rely only on the \textsl{signs} of the derivatives.  These can be usually be determined more easily than the actual expectation values: E.g., in the case when the observable takes only two values, estimating the sign of the derivative amounts to a marjority vote --- which has favourable probabilitistic properties compared to estimating the actual expectation values separately.  Whether the advantage of using the Fourier expansion outweighs the cost of computing it must be tested.

\paragraph{Implementation and computational experiments}
To understand the usefulness of the techniques in practice, the approach has to coded and experimental results have to be performed.

\paragraph{Quantum noise}
Running quantum circuits on NISQ~\cite{Preskill:nisq:2018} processors suffers from limitations such as the low gate fidelities and decoherence, which throw off the expectation values, and make executions of exactly specified gates difficult.  In all the questions above, the effect of quantum noise must be considered.

\paragraph{Understanding the Fourier expansion}
The full $m$-variable Fourier expansion of~$F$ is a trigonometric polynomial with, in general, exponentially many terms (even in the number of qubits), even for 2-eigenvalue Hamiltonians.

Can the $m$-variable Fourier expansion of~$F$ be related to the design of design of the PQC in such a way that, when a family of functions~$F$ which are likely (useful and) difficult to handle on classical computers is identified based on the Fourier expansions, a corresponding family of PQCs can be designed which can compute these functions?  Can this a allow another demonstration of quantum advantage with shallow circuits (cf.\ \cite{Bravyi-Gosset-Koenig:qadvantage:2018})?

\subsection*{Conclusion}
The methods of this paper extend the set of cases where simple evaluations of the parameterized quantum circuits suffice to understand functions it computes.  The number of evaluations grows with the number if differences between eigenvalues -- which can obviously be quadratic in the number of eigenvalues.  In comparison, the method in~\cite{Schuld-Bergholm-Gogolin-Izaac-Killoran:ana-grad:2018} requires only~8 evaluations to determine the gradient, but the circuit which is run is more complex than and significantly different from and the original one.  Moreover, our method allows to determine the local dependence on the coordinate completely, which opens the door to new types of ``training'' algorithms.  Which method should be preferred will depend on the in use of the parameterized quantum circuit.

\section*{Acknowledgements}
This research was supported by the Estonian Research Council, ETAG (\textit{Eesti Teadusagentuur}), through PUT Exploratory Grant \#620.

%%%%%%%%%%%%%%%%%%%%%%%%%%%%%%%%%%%%%%%%%%%%%%%%%%%%%%%%%%%%%%%%%%%%%%%%%%%%%%%%%%%%%%%%%%%%%%%%%%%%%%%%%%%%%%%%
\bibliographystyle{plain}
\bibliography{dirks}
%%%%%%%%%%%%%%%%%%%%%%%%%%%%%%%%%%%%%%%%%%%%%%%%%%%%%%%%%%%%%%%%%%%%%%%%%%%%%%%%%%%%%%%%%%%%%%%%%%%%%%%%%%%%%%%%
\appendix
\section*{APPENDIX}
\section{Completion of the proof of  Theorem~\ref{thm:|D|-evals}\ref{thm:|D|-evals:random}}\label{apx:analytic-sets}
\begin{lemma}\label{lem:zeros-of-analytic-functions}
  Suppose that $h\colon[0,2\pi]^S \to \RR$ is a real analytic function which vanishes on a set of positive Lebesgue measure.  Then~$h$ vanishes identically on $[0,2\pi]^S$.
\end{lemma}
There are many proofs of this fact; we sketch an elementary for the sake of completeness.
\begin{proof}
  Denote by $Z\subseteq[0,2\pi]^S$ a set of positive Lebesgue measure on which~$h$ vanishes.

  Denoting $t' := (t_2,\dots,t_S)$, since
  \begin{equation*}
    0 < \int 1_Z(t) \,dt = \int\int 1_Z(t) \,dt_1\,dt',
  \end{equation*}
  there exists a set $Z'$ of $(S-1)$-dimensional Lebesgue measure such that for all $t'\in Z'$ we have $\int 1_Z(t)\,dt > 0$.  Fix $t'\in Z'$.  There is a set $Z_1(t')$ of non-zero one-dimensional Lebesgue measure such that the analytic function $h(\cdot,t')\colon t_1\mapsto h(t)$ vanishes identically on $Z_1(t')$, and the Principle of Permanence implies that~$h(\cdot,t')$ vanishes identically on~$[0,2\pi]$.  Dealing inductively with $t_2,\dots,t_S$ (see \cite{GilVidal:PhD:2020} for the details), we find that $h$ is identically zero on $[0,2\pi]^S$.
\end{proof}

\end{document}